\documentclass[a4paper,UKenglish,cleveref, autoref, thm-restate]{lipics-v2021}
\usepackage[ruled, noend]{algorithm2e}
\usepackage{mathtools}

\hideLIPIcs  

\DeclareMathOperator*{\argmax}{argmax}

\newcommand{\subtour}{\psi}
\newcommand{\WTSP}{\textsc{Wtsp}}
\newcommand{\QTSP}{\textsc{Qtsp}}

\newcommand{\cost}{\text{cost}}
\newcommand{\cthief}{\textsc{Cttp}}
\newcommand{\bithief}{\textsc{Bttp}}
\newcommand{\cop}{\textsc{Cop}}
\newcommand{\knapdead}{\textsc{Mpbkp}}
\newcommand{\lborder}{s}
\newcommand{\NP}{\mathsf{NP}}
\newcommand{\PNP}{\mathsf{P}}
\title{Approximation Algorithms for the Traveling Thief Problem} 

\author{Jan Eube}{University of Bonn, Germany }{eube@cs.uni-bonn.de}{https://orcid.org/0009-0000-7487-3187}{}

\author{Kelin Luo}{University at Buffalo, NY, USA}{kelinluo@buffalo.edu}{https://orcid.org/0000-0003-2006-0601}{}

\author{Heiko Röglin}{University of Bonn, Germany}{roeglin@cs.uni-bonn.de}{https://orcid.org/0009-0006-8438-3986}{}

\author{Sarah Sturm}{University of Bonn, Germany}{ssturm@uni-bonn.de}{https://orcid.org/0009-0003-1962-8666}{}

\funding{This work has been funded by the Deutsche
Forschungsgemeinschaft (DFG, German Research Foundation) – 390685813;
459420781 and by the Lamarr Institute for Machine Learning and
Artificial Intelligence lamarr-institute.org.}

\authorrunning{J. Eube, K. Luo, H. Röglin and S. Sturm} 

\Copyright{Jan Eube, Kelin Luo, Heiko Röglin and Sarah Sturm} 



\category{} 

\relatedversion{} 

\nolinenumbers 

\EventEditors{John Q. Open and Joan R. Access}
\EventNoEds{2}
\EventLongTitle{42nd Conference on Very Important Topics (CVIT 2016)}
\EventShortTitle{CVIT 2016}
\EventAcronym{CVIT}
\EventYear{2016}
\EventDate{December 24--27, 2016}
\EventLocation{Little Whinging, United Kingdom}
\EventLogo{}
\SeriesVolume{42}
\ArticleNo{23}

\begin{document}

\maketitle

\begin{abstract}
The Traveling Thief Problem (TTP) combines the Traveling Salesperson Problem with the Knapsack Problem. In this problem, a finite metric space is given, and at each location an item with some profit and weight is placed. An agent seeks to collect a subset of the items. To do so, the agent must decide which items to collect and to determine a cyclic tour visiting the corresponding locations.     While collecting an item yields its profit as a reward, the agent’s speed decreases as more weight is picked up. The problem involves two competing objectives: maximizing the total profit of the collected items and minimizing the travel time of the tour. 

While many heuristics and exact algorithms (with a non-polynomial running time) have been developed, no approximation algorithms are known for any variant of the TTP. We aim at computing an $(\alpha_1,\alpha_2)$-approximate Pareto set that, for every solution, contains another solution collecting at least a $\frac{1}{\alpha_1}$ fraction of its profit while requiring at most $\alpha_2$ times its travel time. Our main result is an algorithm that calculates a $(9 + \epsilon,9 + \epsilon)$-approximate Pareto set in polynomial time. 

We also consider the setting in which the set of items to be collected is given in advance, so that the agent only has to compute a tour through the corresponding locations that minimizes the total travel time. This is the so-called Weighted TSP. For this setting, we present a $(2e + \epsilon)$-approximation algorithm.
\end{abstract}

\section{Introduction}

The Knapsack Problem (KP) and the Traveling Salesperson Problem (TSP) are well-studied combinatorial optimization problems that have applications in many different fields. Both problems are known to be $\NP$-hard. In order to deal with them, one either needs to accept a non-polynomial worst-case running time or rely on heuristics or approximation algorithms. While both problems are already challenging on their own, many real-world applications involve the combination of multiple $\NP$-hard optimization problems. These combinations are often referred to as \emph{multi-component optimization problems}. In particular, there is a sequence of papers on a combination of the KP and the TSP called the \emph{Traveling Thief Problem (TTP)} (e.g.,~\cite{Bonyadi2013TTP, DBLP:conf/gecco/PolyakovskiyB0MN14,DBLP:journals/ec/PrzybylekWM18,DBLP:journals/swevo/HerringKY24,DBLP:journals/telo/NikfarjamNN24}). Following its introduction, there were also multiple competitions that aimed at solving variants of this problem\footnote{\url{https://cs.adelaide.edu.au/~optlog/research/combinatorial.php}}.

As in the TSP, also in the TTP one is given a set of locations that have to be visited by an agent (the thief) in a cyclic tour. Additionally, there are items placed at the locations and the agent has to decide which items to pick up. Each item has a profit and a weight and the agent gets slower when it picks up additional weight. The inverse speed is modeled as an increasing function, and the time needed to traverse an edge is obtained by multiplying the current function value by the edge length. There also exists a weight threshold limiting the total weight of the items that the agent is allowed to pick up. Most models in the literature consider the special case in which the agent starts with an initial speed $v_{\max}$ that decreases proportionally to the weight that has already been picked up until a minimum speed $v_{\min}$ is reached when the maximum allowed weight has been picked up. The TTP can be used to model collection and delivery settings in which routing decisions are coupled with the selection of items to be transported.

A solution of the TTP consists of a tour as well as a packing plan that determines which items are picked up in which city. One aims at both maximizing the profit of the collected items and minimizing the time necessary to traverse the tour. Given that these two objectives oppose each other, one usually searches for a trade-off between them. In the literature, often a variant is considered where a scaled version of the travel time is subtracted from the total profit of the collected items, and one aims at maximizing this difference \cite{Bonyadi2013TTP, DBLP:conf/gecco/PolyakovskiyB0MN14,DBLP:journals/ec/PrzybylekWM18,DBLP:journals/telo/NikfarjamNN24}. 

Alternatively, one can formulate the TTP as a bi-objective optimization problem in which the goal is to find the set of Pareto-optimal solutions or a subset thereof that contains a good solution for every possible linear combination of the two objective values \cite{blank2017solving, CHAGAS2022105560}. 
This viewpoint has attracted increasing attention in the evolutionary computation community. In particular, Yafrani et al.\ \cite{yafrani2017multi} highlight the connection between the bi-objective TTP and single objective TTP, and Wu et al.\ \cite{wu2018evolutionary} develop an evolutionary algorithm augmented with a dynamic programming subroutine. 
There is also a variant of the bi-objective TTP in which the profit of an item decreases while it is being carried~\cite{Bonyadi2013TTP}. 
Since both maximizing the profit and minimizing the travel time are $\NP$-hard, it is already $\NP$-hard to decide whether there exists a solution in which the difference of the collected profit and the travel time is non-negative. Consequently, one cannot expect to obtain any polynomial time approximation algorithm for this objective function, unless $\PNP = \NP$. 
For this reason, we focus in this work on approximating a TTP variant that is strongly related to the bi-objective TTP.

The current literature on approximation algorithms for the TTP is rather sparse. There exists a PTAS for the case that the tour of the agent is already fixed and one only needs to determine the packing plan \cite{DBLP:conf/algocloud/NeumannPSSW18}. For the so-called \emph{Weighted TSP} in which the packing plan is already fixed and one only needs to determine the tour, Bossek et al.~\cite{DBLP:conf/gecco/BossekCK020} provide a $3.59$-approximation algorithm if the cost of traversing an edge grows linearly with the weight that has been collected. For more general cost functions, Eube et al.\ \cite{eube2026effectivetravelingmetricinstances} show that the Weighted TSP is polynomially solvable if the metric is a line metric. They also show that the problem is $\NP$-hard even on metrics induced by star-graphs and provide an $8$-approximation for this case. To the authors' knowledge, there exist no approximation results if one needs to determine both the tour and the packing plan, prior to this work.

While our current theoretical knowledge regarding the TTP is rather limited, there exists a rich body of literature for the broader field of vehicle routing problems. In the classic orienteering problem, one is given an undirected graph $G = (V,E)$ with metric edge weights and node profits, a length limit, as well as a start node $s$ and end node $t$. The objective is to find a path from $s$ to $t$ of length at most the given limit that maximizes the profit. There exists a $(2 + \epsilon)$-approximation algorithm for this problem in general metrics \cite{orienteering2apx}.
There exist many variations of this problem. For example, Xu et al.\ give a constant-factor approximation for the team orienteering problem, in which multiple vehicles may be used to visit the nodes~\cite{teamorienteering}. The capacitated orienteering problem (\cop) is probably the variant most closely related to the TTP \cite{bock2015capacitated}. In this problem, nodes also have weights that are independent of the profits, and an additional side constraint requires the total weight of the visited nodes to remain below a given threshold. If we consider the TTP with a constant speed function, the problem basically becomes the $\cop$ problem. Bock and Sanit{\`a} give a $(3 + \epsilon)$-approximation for this problem~\cite{bock2015capacitated}, which we will use in our approximation algorithm for the TTP.

Also relevant to our algorithms is a variant of the TSP known as Quota TSP ($\QTSP$) \cite{ausiello2018prize}. In this problem, nodes are weighted and the objective is to find a tour of minimum length whose total weight meets a given lower bound $k$. If all node weights are one, this is the problem to find a tour of minimum length that visits at least $k$ nodes. This problem is known as the $k$-TSP~\cite{arora20062+}, and there exists a $(2 + \epsilon)$-approximation algorithm for it~\cite{arora20062+}. 

\section{Our Model}

We consider a variant of the TTP that differs slightly from the existing literature. Most changes are introduced to simplify writing and do not functionally change the problem. We will discuss these changes after introducing our model.

We are given a node set $V$ with $|V| = n$, a root $r\in V$ and a distance metric $d:V^2 \rightarrow \mathbb{R}_{\geq 0}$. At every node $v \in V\setminus\{r\}$, an item with a known weight and profit is placed. The weights of the items are given by a function $w: V \rightarrow \mathbb{N}_{ 0}$ and the profits by $p: V \rightarrow \mathbb{N}$. We assume that no item is placed at $r$ itself, which is why we set $w(r) = p(r) = 0$. For simplicity, we also assume that the distance between $r$ and its closest location is exactly $1$, which can be ensured via scaling, unless another location is placed at the same position as $r$. We can model the situation that multiple items are placed at the same position by creating multiple nodes which are at distance $0$ to each other.

There is an agent starting at $r$ who wants to collect the items. To collect an item, the agent must travel to its corresponding location, and we assume that the agent visits only those locations at which items are picked up. Thus, no separate packing plan is needed. In the end, the agent has to return to the initial location $r$. Collecting an item yields its profit as a reward, but requires the agent to carry the item for the remainder of the tour. As the agent accumulates weight, their speed decreases. More formally, let $f: \mathbb{N}_0 \rightarrow \mathbb{R}_{\geq 0}$ be a non-decreasing function modeling the inverse speed. When the agent carries weight $W \in \mathbb{N}_{\geq 0}$, then traversing an edge between two vertices $u,v \in V$ requires time $f(W) \cdot d(u,v)$. The agent needs to follow a tour $\pi = (\pi_1,\ldots,\pi_{m+1})$ with $m \in [n]$ such that $\pi_1 = \pi_{m+1} = r$ and every node in $V \setminus \{r\}$ appears at most once on this tour. The entire profit collected on this tour is equal to $p(\pi) = \sum_{i=2}^{m} p(\pi_i)$. After visiting the $i$-th location on this tour (for an arbitrary $i \in [m]$), the agent has collected weight:
\begin{equation*}
    W_i^\pi = \sum_{j=2}^i w(\pi_j).
\end{equation*}
There is a maximum capacity $w_{\max}$ limiting the total weight that the agent can carry, and we require $w(\pi) := \sum_{i = 2}^m w(\pi_i) \leq w_{\max}$. Without loss of generality we may assume that $f(W) = \infty$ if $W > w_{\max}$. The travel time of the agent to complete such a tour is equal to:
\begin{equation*}
    \cost(\pi) =  \sum_{i=1}^{m} d(\pi_i,\pi_{i+1}) f(W_i^\pi).
\end{equation*}

We introduce the \emph{Constrained Traveling Thief Problem} ($\cthief$). Similar to the Quota TSP, we assume that we are given a desired profit $\mathcal{P}$ and then aim at finding a tour that collects profit at least $\mathcal{P}$ while minimizing the travel time. Unfortunately, it is already $\NP$-hard to decide whether there exists any tour that collects profit at least $\mathcal{P}$, as this would require to determine if there exists a subset of the locations whose entire weight is at most $w_{\max}$ and whose profit is at least $\mathcal{P}$. This is exactly the knapsack problem. To circumvent this, we design bicriteria approximation algorithms. For two constants $\alpha_1, \alpha_2 \geq 1$, a bicriteria $(\alpha_1, \alpha_2)$-approximation algorithm for $\cthief$ calculates a tour $\pi$ such that $\alpha_1 \cdot p(\pi) \geq \mathcal{P} $ and $\cost(\pi) \leq \alpha_2 \cdot \cost(\pi^*)$, where $\pi^*$ denotes a tour with minimum travel time among all tours that collect profit at least $\mathcal{P}$. Figure~\ref{fig:example_tours} depicts an example $\cthief$ instance.

\begin{figure}
\begin{subfigure}{0.4\textwidth}
    \includegraphics[width=\textwidth, page =1]{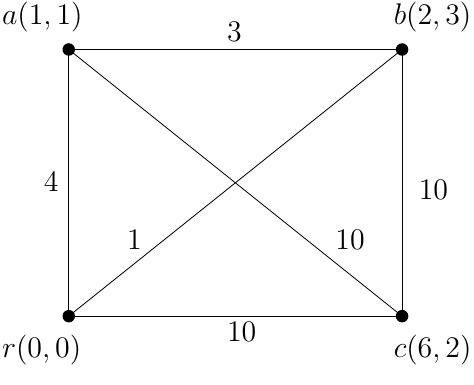}
\end{subfigure}
\hfill
\begin{subfigure}{0.4\textwidth}
    \includegraphics[width=\textwidth, page =2]{images/example_tours.pdf}
\end{subfigure}
\caption{The left figure depicts a TTP instance. The edges are labeled with the respective distances and every node $v$ is followed by the pair $(p(v),w(v))$. We consider the function $f(w) = w+1$ and $w_{\max} = 6$. The right figure depicts the tours $(r,a,b,r)$ in blue and $(r,c,r)$ in green. It holds that $p(r,a,b,r) = 3$ and $\cost(r,a,b,r) = 1 \cdot 4 + 2 \cdot 3+ 5 \cdot 1 = 15$. At the same time, $p(r,c,r) = 6$ and $\cost(r,c,r) = 1 \cdot 10 + 3 \cdot 10 = 40$.}
\label{fig:example_tours}
\end{figure}

While the $\cthief$ problem has not been studied before, it is strongly related with the \emph{Bi-objective Traveling Thief Problem} $(\bithief$)~\cite{blank2017solving, CHAGAS2022105560}. In $\bithief$ one aims at calculating a Pareto set with respect to the profit and the travel time of the tours. We argue in Appendix~\ref{apx:bittp} that one can use a bicriteria $(\alpha_1,\alpha_2)$-approximation algorithm for $\cthief$ to calculate an $((1 + \epsilon) \cdot \alpha_1, \alpha_2)$-approximate Pareto set for $\bithief$.

Most papers on the TTP assume that every location has to be visited, even if the agent does not pick up any item there. The variant in which only a subset of the nodes needs to be visited is more in line with Orienteering and related problems. Furthermore, forcing the agent to visit locations from which no item is collected seems somewhat artificial. In Appendix~\ref{apx:all_location} we argue that our results also carry over to the case that every location needs to be visited.
Our model further assumes that the agent's starting position is fixed in advance and that no item is located at the corresponding node. If the initial position is not given, we can simply evaluate all $n$ choices of $r$ and choose the best one. We can also explicitly deal with the situation that items are placed at this starting position by modifying $r$ and $w_{\max}$. The details of how our algorithms need to be modified are described in Appendix~\ref{apx:no_start} and \ref{apx:mult_items}.

Most papers that deal with TTP also require that $f$ corresponds to a linear speed decrease. 
We consider the more general setting, where $f$ can be an arbitrary non-decreasing function.

We also consider a simplification of the $\cthief$. In the \emph{Weighted TSP} ($\WTSP$), the packing plan is already fixed. Then we can remove a vertex from $V$ if the respective item does not get collected. The problem becomes to find a tour $\pi$ that visits every vertex and minimizes $\cost(\pi)$. Neither the profit function $p$, nor the maximum weight $w_{\max}$ are still necessary.

\section{Our Results}

In Section~\ref{sec:weighted} we prove the following result.

\begin{restatable}{theorem}{ThmWeighted}\label{thm:Weighted}
For any $\epsilon > 0$, one can calculate a $(2e + \epsilon)$-approximation for $\WTSP$ in polynomial time  if all weights are polynomially bounded integers. 
\end{restatable}

The restriction to instances with polynomially bounded integer weights stems from the $(2+\epsilon)$-approximation algorithm for the $k$-TSP, which we use as a subroutine in our algorithm~\cite{arora20062+}. This algorithms is stated for the case that all nodes have unit weight, but it can easily be extended to polynomially bounded integers weights by replacing each node with multiple copies. Moreover, we have examined the algorithm from~\cite{arora20062+} carefully and are confident that it also applies, without substantial modifications, to the case of general integer weights. If this is indeed the case, then Theorem~\ref{thm:Weighted} also extends to arbitrary integer weights.

In Section~\ref{sec:thief} we then present our main result, which is a constant bicriteria approximation algorithm for the Constrained Traveling Thief problem (\cthief).

\begin{restatable}{theorem}{ThmThief}\label{thm:Thief}
For any $\epsilon > 0$ there exists an polynomial time bicriteria $(9 + \epsilon, 9 + \epsilon)$-approximation algorithm for $\cthief$. 
\end{restatable}

\section{Algorithm for a Given Packing Plan}\label{sec:weighted}

For simplicity, we only describe an $(8 + \epsilon)$-approximation algorithm for $\WTSP$ in this section. In Appendix~\ref{apx:improved_weighted} we explain how the approximation factor can be improved to $(2 e + \epsilon)$.

Our algorithm for $\WTSP$ follows the general structure of the doubling algorithm presented by Epstein et al. \cite{epstein2010universal}. We consider the segments of tours that the agent travels after picking up at least some specific weight. 

\begin{definition}
    For a TSP-tour $\pi$ we define $D^{\pi}(W)$ to be the distance the agent travels after picking up weight at least $W \in \mathbb{N}_0$, i.e., let $i$ be the smallest index such that $W_i^\pi \geq W$ then we define:
    \begin{equation*}
        D^{\pi}(W) =  \sum_{j=i}^{n} d(\pi_j,\pi_{j+1}).
    \end{equation*}
\end{definition}

Intuitively, if for every possible weight the distance that the agent travels after picking up that weight is bounded by a constant multiple of the corresponding distance in an optimal solution, then the total travel cost is bounded by the same constant multiple of the optimum value. Indeed, this can be shown by adapting the arguments from \cite{epstein2010universal}.

\begin{lemma}\label{lem:bound}
    Let $\pi$ be a $\WTSP$ tour and let $\pi^*$ be the optimum $\WTSP$ tour. Let $c$ be an arbitrary constant. If it holds for any weight $W$ that $D^{\pi}(W) \leq c \cdot D^{\pi^*}(W)$ then $cost(\pi) \leq c \cdot cost(\pi^*)$.
\end{lemma}

\begin{proof}
    For simplicity we define $f(-1) = 0$. For any tour $\pi$ we can express $\cost(\pi)$ as follows:
    \begin{equation*}
        \cost(\pi) = \sum_{W= 0}^{w(V)} D^{\pi^*}(W) \cdot (f(W) - f(W-1))
    \end{equation*}
    Therefore we obtain:
    \begin{align*}
        \cost(\pi) &=  \sum_{W= 0}^{w(V)} D^\pi(W) \cdot (f(W) - f(W-1))\\
        &\leq \sum_{W= 0}^{w(V)} c \cdot D^{\pi^*}(W) \cdot (f(W) - f(W-1))\\
        &= c \cdot \cost(\pi^*)\qedhere
    \end{align*}
\end{proof}

To find a tour $\pi$ that fulfills the requirement of this lemma, we basically need to obtain for a given weight $\mathcal{W}$ a subtour $\subtour$ whose nodes have a total weight of at least $\mathcal{W}$ and whose length is bounded by $c \cdot D^{\pi^*}(w(V) - \mathcal{W})$ for a constant $c$. If we append $\subtour$ at the end of $\pi$ and do not pick up the items placed on $\psi$ during any previous subtour, we guarantee that the agent has picked up weight at most $w(V) - \mathcal{W}$ before starting with $\subtour$. Thus, the agent does not travel much more distance with this weight than in the optimum tour. To achieve this, we use an approximation algorithm for the so-called \emph{Quota TSP} \cite{ausiello2018prize}.

\begin{definition}
    Given a metric $(V,d)$, a root $r \in V$, a weight function $w:V \rightarrow \mathbb{N}$ and a number $W^* \in \mathbb{N}$. The Quota TSP ($\QTSP$) asks for a sequence $\psi = \psi_1,...,\psi_l, \psi_{l+1}$ of vertices in $V$ with $\psi_1 = \psi_{l+1}= r$ such that the total weight of the vertices on this sequence is at least $W^*$ and the length of the cycle visiting all the nodes in the given order is minimized. More formally we want to find a $\psi$ such that:
    \begin{itemize}
        \item $\sum_{i=1}^l w(\psi_i) \geq W^*$
        \item $\ell(\psi) :=  \sum_{i=1}^{l} d(\psi_i,\psi_{i+1})$ gets minimized.
    \end{itemize}
\end{definition}

We use a subroutine for $\QTSP$ that is based on an algorithm due to Chaudhuri et al.~\cite{chaudhuri2003paths}. They proved that in the special case where every vertex has unit weight (so that one simply counts the number of vertices) one can, for any given $k$ and root $r$, compute in polynomial time a tree that contains $r$, spans at least $k$ vertices, and has total edge length at most $(1+ \epsilon)$ times the length of the shortest path ending at $r$ that visits at least $k$ vertices. While Chaudhuri et al. do not address the case in which vertices have integer weights, their arguments carry over to this generalized case. Thus, one can also calculate for a desired weight $W^*$ a tree containing $r$ that spans weight at least $W^*$ with length bounded by $(1+\epsilon)$ times the length of any path ending in $r$ that picks up weight at least $W^*$. Since this tree can be transformed into a tour (and thus a valid $\QTSP$ solution) by doubling all its edges and then skipping vertices that are visited repeatedly we obtain the following theorem.

\begin{theorem}
    For any $\epsilon > 0$, there exists a polynomial-time algorithm that given a root $r$ and a constant $W^*$ finds a $\QTSP$ tour starting and ending in $r$ whose length is upper bounded by $(2+ \epsilon)$ times the length of any path ending in $r$ with total weight at least $W^*$.
\end{theorem}

The fact that we compare the quality of the tour against a path ending in $r$ is relevant because such paths include the subpaths of the optimum tour after a specific weight has been picked up and thus are directly related to the values $ D^{\pi^*}(w(V) - \mathcal{W})$ for some given weight~$\mathcal{W}$. In our algorithm, we use a subroutine $\textsc{Qtsp}(v,(R,d),W)$ that returns a $\QTSP$ solution within this cost bound with root $v$ and weight parameter $W$ that only visits nodes in $R \subseteq V$ (one may note that $d$ restricted to $R$ is again a metric). We denote the length of this tour as $\ell(\textsc{Qtsp}(v,(R,d),W))$. We define $\ell(\textsc{Qtsp}(v,(R,d),W)) = \infty$ if there exists no feasible $\QTSP$ tour because $W$ is larger than the sum of the weights of the items in $R$.

The idea is now to find a sequence of subtours of increasing lengths (corresponding to increasing weights) and letting the agent travel over them in reversed order, i.e., starting with the longest tour and ending with the shortest one. To avoid that multiple subtours contain the same vertex, we maintain a set $R$ of vertices that are not contained in any tour yet and restrict the larger tours to only these vertices. We are able to ensure that for every weight~$\mathcal{W}$ there exists a subtour such that the agent does not pick up weight $\mathcal{W}$ before starting with this tour and the length of it as well as all subsequent tours is upper bounded by a constant times $D^{\pi^*}(\mathcal{W})$. This results in Algorithm \ref{Alg:8_apx}.

\begin{algorithm}
	\LinesNumbered
	\DontPrintSemicolon
	\SetKwInOut{Input}{input}
	\SetKwInOut{Output}{output}
    \caption{An $(8 + \epsilon)$-approximation for $\WTSP$.}\label{Alg:8_apx}
    $R = V$, $s = 0$\;
    \While{ $R \neq \{r\}$}{
        $s = s+1$\;
        Find $w_s$ s.t.\ $\ell(\textsc{Qtsp}(r, (R,d), w_s)) \leq 2^s < \ell(\textsc{Qtsp}(r,(R,d), w_{s}+1))$ (binary search)\;
        $\subtour_s = \textsc{Qtsp}(r,(R,d),w_s)$\;
        Remove all vertices in $\subtour_s$ (except $r$) from $R$\;
    }
    Initialize $\pi$ as a sequence only containing $r$\;
    \For{$i=s,\ldots 1$}{
     \If{$\subtour_i \neq (r)$}{Append $\subtour_i$ to $\pi$ (ignoring the vertex $r$)\;}
    }
    \textbf{return:} $\pi$
\end{algorithm}

To analyze the approximation ratio of this algorithm, we will often be interested in the total weight picked up in the $i$ smallest tours for a given $i \in [s]$. We define:

\begin{equation*}
    W_i^\subtour := \sum_{j=1}^i w_j.
\end{equation*}

For any weight $\mathcal{W}$, the path that the agent travels after picking up weight at least $\mathcal{W}$ ends in $r$ and contains vertices with weight at least $w(V) - \mathcal{W} + 1$. Given that the $\QTSP$ subroutine compares against the lengths of such paths, we can lower bound the value $D^{\pi^*}(\mathcal{W})$ as follows.

\begin{lemma}\label{lem:opt_plan_lower}
    For any $\mathcal{W}$ and $i \in [s]$ with $w(V) - \mathcal{W} \geq W_i^\subtour$ it holds that $2^{i} \leq (2 + \epsilon) D^{\pi^*}(\mathcal{W})$.
\end{lemma}

\begin{proof}

    Let $i$ be an index such that $2^{i} > (2 + \epsilon) D^{\pi^*}(\mathcal{W})$. We will show that $W_i^\subtour> w(V) - \mathcal{W}$. Let $l$ be the index after which the agent has picked up weight at least $\mathcal{W}$ on the optimum tour $\pi^*$. Then the path $\pi^*_l,\pi^*_{l+1},\ldots,\pi^*_n,\pi^*_1$ has exactly length $D^{\pi^*}(\mathcal{W})$. The total weight of the vertices on this path is at least $w(V)-\mathcal{W} + 1 $, given that the total weight of the previous vertices is strictly smaller than $\mathcal{W}$ and thus upper bounded by $\mathcal{W} - 1$.

     Let $R$ be the set of vertices left in the graph at the beginning of the iteration of Algorithm~\ref{Alg:8_apx} in which $\subtour_i$ gets calculated. Let $\subtour^*$ be the path that gets created by removing all vertices in $V \setminus R$ from $\pi^*_l,\pi^*_{l+1},\ldots,\pi^*_n,\pi^*_1$. Since the total weight of items in $V \setminus R$ is at most $W_{i-1}^\subtour$, the weight of the vertices on this path is at least $w(V)-\mathcal{W} - W_{i-1}^\subtour + 1$. Additionally by the triangle inequality we know that $\ell(\subtour^*) \leq D^{\pi^*}(\mathcal{W})$.

    For any weight $W$, the subroutine  $\textsc{Qtsp}(r, (R,d), W)$ returns a tour with length at most $(2 + \epsilon)$ times the length of the shortest path ending in $r$ that picks up at least the desired weight $W$. Thus, we know that for any $W \leq w(V)-\mathcal{W} - W_{i-1}^\subtour + 1$ the length of the resulting tour is bounded by $(2 + \epsilon) \ell(\subtour^*)\leq (2 + \epsilon) D^{\pi^*}(\mathcal{W}) < 2^{i}$ and the algorithm chooses $w_i \geq w(V)-\mathcal{W} - W_{i-1}^\subtour + 1 $. This clearly implies that $W_i^\subtour = w_i + W_{i-1}^\subtour \geq w(V) - \mathcal{W} +1$ which proves the lemma.
\end{proof}

At the same time, we can bound $D^{\pi}(\mathcal{W})$ by the length of the $i + 1$ shortest subtours for a given index $i \in [s]$ if the total weight $W_{i+1}^\subtour$ of these tours is larger than $w(V) - \mathcal{W}$.

\begin{lemma}\label{lem:alg_plan_upper}
For any $\mathcal{W} \in \mathbb{N}_0$ and any $i \in [s]$ with $w(V) - \mathcal{W} \leq W_{i+1}^\subtour - 1$ it holds that $D^\pi(\mathcal{W})\leq 2^{i+2}$
\end{lemma}

\begin{proof}
    One may note that the entirety of vertices in the tours $\subtour_{i+1}, \ldots ,\subtour_1$ have a total weight of at least $W_{i+1}^\subtour$. Thus, the weight that $\pi$ picked up before visiting any of these vertices is bounded by $w(V) - (W_{i+1}^\subtour) \leq \mathcal{W} - 1$. As a result $D^{\pi}(\mathcal{W})$ is upper bounded by the remaining length of $\pi$ after the first vertex (except $r$) on one of these tours is visited. Given that we remove a vertex from $R$ if it appears in a tour, the tours $\subtour_{i+2},\ldots,\subtour_n$ do not contain any vertices in $\subtour_{i+1}, \ldots ,\subtour_1$. Thus, we can bound $D^{\pi}(\mathcal{W})$ by the length of the combined subtours $\subtour_{i+1}, \ldots ,\subtour_1$:
    \begin{equation*}
        D^{\pi}(\mathcal{W}) \leq \sum_{j=1}^{i+1} \ell(\subtour_j) \leq \sum_{j=1}^{i+1} 2^{j} \leq 2^{i+2}\qedhere
    \end{equation*}
\end{proof}

By inserting the previous two lemmas into Lemma \ref{lem:bound} we obtain.

\begin{theorem}
    If all weights are polynomially bounded integers, Algorithm \ref{Alg:8_apx} calculates an $(8 + 4 \epsilon)$-Approximation of the optimum $\WTSP$ tour in polynomial time.
\end{theorem}

\begin{proof}
    First, we prove the approximation factor. By Lemma \ref{lem:bound}, it is sufficient to show that for any $\mathcal{W} \in [w(V)]$ it holds that $D^{\pi}(\mathcal{W}) \leq (8 + 4 \epsilon) D^{\pi^*}(\mathcal{W})$. Let $\mathcal{W}$ be chosen arbitrarily. Then there exists an index $i \in [s]$ such that $W_i^\subtour \leq w(V) - \mathcal{W} \leq W_{i+1}^\subtour -1$. By Lemma~\ref{lem:opt_plan_lower} and \ref{lem:alg_plan_upper} we obtain:
    \begin{equation*}
        D^\pi(\mathcal{W}) \leq 2^{i+2} \leq 4 (2+ \epsilon) D^{\pi^*}(\mathcal{W}) = (8 + 4 \epsilon)D^{\pi^*}(\mathcal{W}).
    \end{equation*}

    To bound the running time, let $u,v$ be the two vertices maximizing the pairwise distance $d(u,v)$. Then the length of any tour returned by the $\QTSP$ algorithm has at most length $n \cdot d(u,v)$. As a result, we can bound the final value of $s$ and thus also the iterations of the while loop by $\lceil \log(n \cdot d(u,v))\rceil$. Given that our instance requires at least $\Omega(\log(d(u,v)))$ many bits to encode $d(u,v)$, this is bounded polynomially of the input size. 

    Given that we assumed that the weights are integers we also know that we can find for any $i \in [s]$ the value $w_i$ by $O(\log(w(V)))$ uses of the algorithm $\textsc{Qtsp}$ which is again polynomial. Also the length of all other steps of the algorithm clearly only require polynomial time. Thus, the theorem holds.
\end{proof}

\section{A Bicriteria Approximation for Constrained TTP}\label{sec:thief}

In this section we develop a bicriteria approximation algorithm for the Constrained Traveling Thief Problem (\cthief). For this, we need another method to bound the cost of a tour against the cost of an optimum solution. Let us assume w.l.o.g.\ that $f(0) = 1$ (which can be ensured via scaling). We subdivide the range of potential weights $0,\ldots, w_{\max}$ into intervals in which the speed of the agent does not differ by more than a factor of two. To be more precise let $s = \lceil \log(f(w_{\max})\rceil - 1$. Then for every $i \in [s]$ we define a value $T_i$ such that $f(T_i) < 2^i \leq f(T_i +1)$ and for $i = 0$ we set $T_0 = 0$ and for $i = {s+1}$ we set $T_{s+1} = w_{\max}$. Then we can lower bound the cost of the optimum solution as follows.

\begin{figure}
\centering
    \includegraphics[width= 0.8 \textwidth]{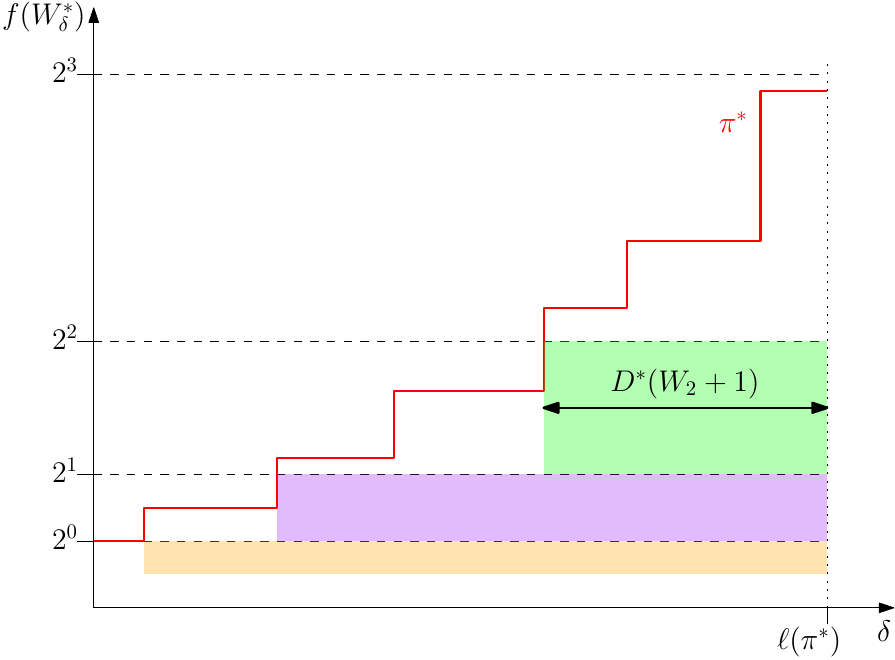}
    \caption{A sketch depicting how the cost of the optimum solution can be lower bounded. For a distance $\delta$ we define $W_\delta^*$ to be the weight that the agent has picked up on the optimum tour $\pi^*$ after traveling distance $\delta$. The the red curve depicts the value of the cost function $f$ over the course of $\pi^*$ and it is easy to verify that the cost of the tour is equal to the area below this curve. In Lemma \ref{lem:alt_lower_bound} we lower bound this by the sum of the areas of the colored rectangles.} \label{fig:lower_bound}
\end{figure}
 
\begin{restatable}{lemma}{LemAltLowerBound}\label{lem:alt_lower_bound}
Let $\pi^*$ be the optimum tour. It holds that
    \begin{equation*}
        \cost(\pi^*) \geq \sum_{i=0}^{s} 2^{i-1} \cdot D^{\pi^*}(T_{i} + 1).
    \end{equation*}
\end{restatable}

\begin{proof}
    Figure \ref{fig:lower_bound} sketches the idea behind this proof. We consider an alternative cost function $g$ with $g(w) = 2^i$ for any $w \in \{T_i + 1, \ldots, T_{i+1}\}$ and denote the cost of tour $\pi^*$ if we use $g$ instead of $f$ as $cost'(\pi^*)$. It is easy to observe that for any $w \in [w_{\max}]\cup \{0\}$ it holds that $f(w) \geq g(w)$ and thus $ cost(\pi^*)\geq cost'(\pi^*)$.
    For any $i \in \{0,\ldots,s-1\}$ we denote the length of the subpath on which the agent has picked up weight at least $T_i + 1$ and less than $T_{i+1} + 1$ by $\ell_i$. One might note that $\ell_i$ might be $0$ if there exists a vertex where the agent arrives with weight smaller $T_i + 1$ and leaves with weight at least $T_{i+1} + 1$. For $i = s$ we define $L_s$ as the length of the subpath after the agent has picked up weight at least $T_s + 1$. Then
    \begin{align*}
        cost'(\pi^*) &\geq \sum_{i=0}^s g(T_{i}+1) \cdot \ell_i\\
        &= \sum_{i=0}^s 2^i \cdot \ell_i\\
        &\geq \sum_{i=0}^s (2^i - 2^{i-1})\cdot \sum_{j=i}^s \ell_j\\
        &= \sum_{i=0}^s 2^{i-1} \cdot D^{\pi^*}(T_{i} + 1),
    \end{align*}
    where we used in the last inequality that $D^{\pi^*}(T_i+1)$ denotes exactly the length of the subtour after the agent has picked up weight at least $T_i +1$ which can be subdivided into the subpaths corresponding to $\ell_1,\ldots, \ell_s$.    
\end{proof}

Intuitively, this lemma will allow us to calculate for every interval $[T_{i} +1,T_{i+1}]$ a subtour on which the collected weight lies within this interval. Given that the speed of the agent only differs by $2$ on this subtour the exact order of the locations visited is not too important. If for every $i$ the length of the subtour is bounded by $D^{\pi^*}(T_{i} + 1)$ we know that if we append these tours after each other the total travel time of the resulting tour is then also bounded by a constant times $\cost(\pi^*)$.

For any $i \in [s] \cup \{0\}$ we define the the index $b_i$ such that $W_{b_i -1}^{\pi^*} \leq T_i$ and $W^{\pi^*}_{b_i} \geq T_i +1$. Or in other words the locations $\pi_{b_1}^*,\ldots,\pi_{b_s}^*$ are exactly those where the weight that the agent picked up moves into another interval. For technical reasons we also define $b_{s+1} = m +1$ (where $m$ is the number of items collected by the optimum tour). We define $B = \{ \pi_{b_i}^*\mid i \in [s] \cup \{0\}\}$. We present two algorithms that compute two possible $\cthief$ tours $\pi^{(1)}$ and $\pi^{(2)}$ such that both $\cost(\pi^{(1)})$ and $\cost(\pi^{(2)})$ can be bounded by a constant times $\cost(\pi^*)$. Additionally, we have that if $p(B)$ is not too large then $p(\pi^{(1)})$ is at least a constant fraction of $\mathcal{P}$. Otherwise, $p(\pi^{(2)})$ can be lower bounded and we obtain a bicriteria approximation by always taking the more profitable tour.

To calculate $\pi^{(1)}$, we consider for any $i \in [s] \cup \{0\}$ the path $(\pi_{b_i +1}^*,\ldots,\pi_{b_{i+1} -1}^*,r)$ that visits all vertices on the optimum tour appearing after $\pi_{b_i}^*$ and before $\pi_{b_{i+1}}^*$. Then the length of this path is upper bounded by $ D^{\pi^*}(T_{i} + 1)$ and the weight picked up by $T_{i+1} - T_i$. The existence of such a path will enable us to find a subtour corresponding to the weight interval $[T_{i} +1,T_{i+1}]$ that is not too expensive (when starting with the weight collected by the subtours corresponding to the previous intervals) and picks up items whose profit is lower bounded by a constant times $\sum_{j= b_i +1}^{b_{i+1} -1} p(\pi_j^*)$. Since every item in $V(\pi^*)$ that is not contained in $B$ is contained in one of this described paths we will be able to use this to lower bound the values of $\pi^{(1)}$. The details are presented in Section \ref{sec:path1}.

To obtain the tour $\pi^{(2)}$, we can use Lemma \ref{lem:alt_lower_bound} to conclude that a tour that visits the items in $B$ in the same order as $\pi^*$ and returns to $r$ between any two items does not take too much longer than $\pi^*$. Since the items are directly visited from $r$, the distance that needs to be traveled to reach an item is independent from the remainder of the tour. As a result, we can use a knapsack algorithm that yields a tour whose profit is at least a constant fraction of $p(B)$. The details are presented in Section \ref{sec:path2}.

For both algorithms we need to have an estimate for the travel time of the optimum tour. We assume that a value $\mathcal{D}$ is given such that $ \cost(\pi^*) \leq \mathcal{D} \leq (1 + \epsilon) \cost(\pi^*)$. This can be ensured by trying out all values $(1+\epsilon)^t$ for $t \leq \lceil \max_{v,w \in V} \log_{(1+ \epsilon)}(n \cdot d(v,w) \cdot f(w_{\max})) \rceil$.  

Additionally, we denote for a subtour $\subtour = (\subtour_1,..,\subtour_m, \subtour_{m+1})$ where $\subtour_1= \subtour_{m+1}$ and a weight $W$ the cost of $\subtour$ when the agent collected already weight $W$ as 
\begin{equation*}
    \cost(\subtour, W) =  \sum_{i=1}^{m} d(\subtour_i,\subtour_{i+1}) f(W_i^\subtour + W).
\end{equation*}
This definition will help us when bounding the cost of the calculated solutions. 

\subsection{Obtaining the First Solution}\label{sec:path1}

To obtain the solution $\pi^{(1)}$ we need some results for the so called \emph{capacitated orienteering problem} (\cop)~\cite{bock2015capacitated}. In this problem one is given a point set $V$, a distance metric $d: V^2 \rightarrow \mathbb{R}_{\geq 0}$, a weight function $w: V \rightarrow \mathbb{N}$, a reward function $p: V \rightarrow \mathbb{N}$, a capacity $C \in \mathbb{N}$, a length restriction $L \in \mathbb{R}_{\geq 0}$ and a root $r \in V$. The objective is to obtain a path $\subtour$ ending in $r$ such that the total weight of the vertices on the path is bounded by $C$, the length is bounded by $L$ and the reward of the vertices on the path gets maximized. Bock and Sanit{\`a}~\cite{bock2015capacitated} showed the following theorem.

\begin{restatable}{theorem}{ThmCOP}
    For any constant $\epsilon > 0$, one can $(3+ \epsilon)$-approximate the the capacitated orienteering problem in polynomial time. If the distance metric $d$ is Euclidean the approximation factor can be improved to $(1+ \epsilon)$.
\end{restatable}

In our algorithm we use a subroutine $\cop(S,C,L)$ that returns a path ending in $r$ that only visits vertices in $S \subseteq V$ such that its weight is at most $C$, its length at most $L$ and its profit is lower bounded by $\frac{1}{3 + \epsilon}$ times the profit of any other path fulfilling this requirements. We directly interpret the result of the subroutine as a subtour on which the agent travels from $r$ to the start of the path and then follows the path back to $r$.

If we knew the value $D^{\pi^*}(T_{i} +1)$ for an $i \in [s] \cup \{0\}$ we could use this algorithm to find a `good' subtour of length at most $D^{\pi^*}(T_{i} +1)$ that picks up at most weight $T_{i+1} - T_i$ which means that the collected weight stays within the interval $[T_{i}, T_{i+1}]$ if the agent collected at most weight $T_i$ on the previous tours. We may `guess' the value of $D^{\pi^*}(T_{i} +1)$ with accuracy $(1+\epsilon)$ by trying out all values $(1+\epsilon)^j$ for reasonable integer values $j$, but it is not directly possible to detect whether or not a guess is correct. For two different guesses of the length, it is possible (and actually likely) that for one guess the resulting tour produced by $\cop$ will take less time to traverse while the tour of the other guess picks up a larger profit. The choice which tour is preferable depends on the knowledge how much the agent needs to travel within the other intervals to pick up a certain profit. 

To circumvent this problem, we may use our estimate $\mathcal{D}$ of the cost of the optimum tour. If we only consider `efficient' subtours on which the ratio between the profit and the cost of the tour is lower bounded by a constant times $\frac{\mathcal{P}}{\mathcal{D}}$ then we know that the cost of the resulting tour can only exceed the cost of the optimum tour if we also pick up a constant fraction of the profit. If the cost of the optimum solution is exceeded by too much, we can simply remove some of the subtours. When considering the possible subtours for an interval, we pick the most profitable among all efficient subtours. We will be able to use Lemma \ref{lem:alt_lower_bound} to ensure that then at least a constant fraction of the optimum profit is picked up (if $p(B)$ is not too large). We end up with Algorithm~\ref{Alg:profit_tour1}.

\begin{restatable}{algorithm}{AlgPiOne}
	\LinesNumbered
	\DontPrintSemicolon
	\SetKwInOut{Input}{input}
	\SetKwInOut{Output}{output}
    \caption{Calculation of the first tour $\pi^{(1)}$.}\label{Alg:profit_tour1}
    $P_{0}=\emptyset$, $t_1 = \frac{1}{9 + 3 \epsilon} \cdot \frac{\mathcal{P}}{(6 + 6\epsilon)\mathcal{D}}$ \;
     \For{ $i = 0, \ldots s$}{
         $\subtour_i = (r)$\;
         \For{$j=0,\ldots,\lceil\log_{1+\epsilon}(\mathcal{D}/2^{i})\rceil$}{
             $\subtour' = \textsc{Cop}((V \setminus P_i,d),T_{i+1} - T_i,(1+\epsilon)^j)$\;
             \If{$\frac{p(\subtour')}{\cost(\subtour',w(P_i))}\geq t_1$}{
                 \If{$p(\subtour')\geq p(\subtour_i)$}{
                     Set $\subtour_i = \subtour'$\;
                 }
             }
         }
         Set $P_{i +1} = P_i \cup V(\subtour_i)$\;
     }
    Initialize $\pi^{(1)}$ as a sequence only containing $r$\;
    $i= 0$\;
    \While{$\cost(\pi^{(1)})\leq (6 + 6\epsilon) \mathcal{D}$ and $i \leq s$}{
     \If{$\subtour_i \neq (r)$}{Append $\subtour_i$ to $\pi^{(1)}$ (ignoring the vertex $r$)\;}
     $i = i+1$\;
    }
    \textbf{return:} $\pi^{(1)}$
\end{restatable}

For any $i \in [s] \cup \{0\}$ we consider the subtour that picks up all items that are collected on $\pi^*$ while the collected weight lies within $[T_i,T_{i+1}]$. If the ratio between the profit and the travel time of this tour (if one starts with weight $T_i$) is not too low we call the interval $[T_i,T_{i+1}]$ efficient. We provide a formal definition that also extends to the case that some items have already been collected:

\begin{restatable}{definition}{DefPiOneEff}
    For any $i \in [s] \cup \{0\}$ we define the set $A_i = \{\pi_j^*\mid T_{i}+1 \leq W_{j-1}^{\pi^*}\leq T_{i+1}-w(\pi^*_j)\}$ of items that are visited on the optimum tour such that both directly before and after this visit the collected weight lies within $[T_{i}+1,T_{i+1}]$. 
    
    For any set $P \subseteq V$ we say that the interval $[T_{i}+1,T_{i+1}]$ is efficient after $P$ has been collected if $\frac{p(A_i \setminus P)}{D^{\pi^*}(T_{i}+1)\cdot 2^{i+1}} \geq \frac{1}{3} \cdot \frac{\mathcal{P}}{ 4 \cdot \mathcal{D}}$.
\end{restatable}

    

If an interval is efficient this guarantees us a good solution that we can compare against when calculating $\subtour_i$:

\begin{restatable}{lemma}{LemEfficientSubtour}\label{lem:efficient_subtour}
    For any $i \in [s] \cup \{0\}$ it holds that if the interval $[T_{i}+1,T_{i+1}]$ is efficient after $P_i$ has been collected then $p(\subtour_i) \geq \frac{p(A_i \setminus P_i)}{3+ \epsilon}$.
\end{restatable}

\begin{proof}
    Let us consider the iteration in which we determine $\subtour_i$. Let $j = \lceil\log_{(1+\epsilon)}(D^{\pi^*}(T_i + 1))\rceil$. Since $f(T_i+1) \geq 2^i$ we know that $D^{\pi^*}(T_i + 1) \leq \frac{\cost(\pi^*)}{2^i} \leq \frac{\mathcal{D}}{2^i} $. Thus, the algorithm will consider $\subtour' = \textsc{Cop}(V \setminus P_i,(1+\epsilon)^j,T_{i+1} - T_i)$ as a possible choice of $\subtour_i$. Let us consider the path that visits the items in $A_i \setminus P_i$ in the same order as $\pi^*$ and then returns to $r$. Then the length of this path is upper bounded by $D^{\pi^*}(T_i + 1)$ and its weight by $T_{i+1} -T_i$. The total profit of this path is $p(A_i \setminus P_i)$, which together with the fact that $\cop$ is a $(3+\epsilon)$-approximation implies that $p(\subtour') \geq \frac{p(A_i \setminus P_i)}{3 + \epsilon}$.

    At the same time the length of the path corresponding to $\subtour'$ is at most $(1+\epsilon) D^{\pi^*}(T_i + 1)$. Since on the subtour the agent first travels from $r$ to the start of this path and then traverses the path itself we may bound $\cost(\subtour',w(P_i))$ as follows:
    \begin{align*}
        \cost(\subtour',w(P_i)) &\leq (1+\epsilon)D^{\pi^*}(T_i + 1) f\big(w(P_i)\big) + (1+\epsilon)D^{\pi^*}(T_i + 1) f\big(w(P_i) + T_{i+1} -T_i\big)\\
        &\leq (1+\epsilon)D^{\pi^*}(T_i + 1) (f(T_i) + f(T_{i+1}))\\
        &\leq (1+\epsilon)D^{\pi^*}(T_i + 1) (2^i + 2^{i+1})\\
        &= 1.5\cdot (1+\epsilon)\cdot\big(D^{\pi^*}(T_i + 1)  \cdot 2^{i+1}\big)\\
        &\leq \frac{3\cdot (6+6\epsilon) \cdot \mathcal{D}\cdot p(A_i \setminus P)}{ \mathcal{P}}.
    \end{align*}
    Thus:
    \begin{equation*}
        \frac{p(\subtour')}{\cost(\subtour',w(P_i))} \geq \frac{p(A_i \setminus P_i)}{3+\epsilon} \cdot \frac{ \mathcal{P}}{ 3 \cdot (6+6\epsilon) \cdot \mathcal{D}\cdot p(A_i \setminus P)} = t_1.
    \end{equation*}
    As a result the algorithm will choose $\subtour_i = \subtour'$ unless it finds another more profitable  subtour and $p(\subtour_i)$ is lower bounded by $p(\subtour') \geq \frac{p(A_i \setminus P_i)}{3+\epsilon}$.
\end{proof}

By bounding the profit that corresponds to the inefficient intervals, we are able to prove that there still exists an efficient interval if the total profit of the already collected items is smaller than $\frac{\mathcal{P}}{9 + 3 \epsilon}$.

\begin{restatable}{lemma}{LemEffItemsLeft}\label{lem:eff_items_left}
    If $p(B) \leq \frac{2 + \epsilon}{9 + 3 \epsilon} \mathcal{P}$ it holds for any $i \in [s] \cup\{0\}$ that if $p(P_i) < \frac{ \mathcal{P}}{9 + 3\epsilon}$ then there exists an $i'\geq i $ such that the interval $[T_{i'}+1,T_{i'+1}]$ is efficient after $P_i$ has been collected.
\end{restatable}

\begin{proof}
    For simplicity we say that an interval is efficient if it is efficient after $P_i$ has been collected. We subdivide the intervals into three categories:
    \begin{itemize}
        \item $I_{inef} = \{j \in [s] \cup \{0\}\mid [T_{j}+1,T_{j+1}]\text{ is not efficient.}\}$ contains all the indices corresponding to inefficient intervals.
        \item $I_{block} := \{j \in [i-1] \cup \{0\}\mid[T_j +1,T_{j+1}]\text{ is efficient.}\}$ contains all indices smaller $i$ corresponding to efficient intervals. We call the respective intervals blocked.
        \item $I_{eff} := \{j \in[s] \cup \{0\}\mid j \geq i\land [T_j +1,T_{j+1}]\text{ is efficient.}\}$ contains all indices with value at least $i$ corresponding to efficient intervals.
    \end{itemize}
    Let $p(P_i)< \frac{ \mathcal{P}}{9 + 3\epsilon}$. Since every item on the optimum tour is contained in $B$ or in exactly one of the intervals we know that:
    \begin{align*}
        \sum_{j \in I_{inef}} p(A_j \setminus P_i) + \sum_{j \in I_{block}} p(A_j \setminus P_i) + \sum_{j \in I_{eff}} p(A_j \setminus P_i) &= \sum_{j=0}^s p(A_j \setminus P_i)\\
        &\geq \left(\sum_{j=0}^s p(A_j)\right) - p(P_i)\\
        &> \frac{7 + 2\epsilon}{9 + 3 \epsilon} \mathcal{P} - \frac{ \mathcal{P}}{9 + 3\epsilon}\\
        &= \frac{2}{3} \mathcal{P}
    \end{align*}
    
    We want to bound the contribution of the inefficient and the blocked intervals to the left side of the inequality. For any $j \in I_{inef}$ it holds that:
    \begin{equation*}
        p(A_j \setminus P_i) < \frac{\mathcal{P}}{12 \mathcal{D}} \cdot D^{\pi^*}(T_j +1) 2^{j+1} = \frac{\mathcal{P}}{3\mathcal{D}} D^{\pi^*}(T_j +1) 2^{j-1}
    \end{equation*}
    By combining this with Lemma~\ref{lem:alt_lower_bound} we obtain:
    \begin{align*}
        \sum_{j \in I_{inef}} p(A_j \setminus P_i) &\stackrel{\phantom{\text{Lem. } \ref{lem:alt_lower_bound}}}{<} \frac{\mathcal{P}}{3 \mathcal{D}} \sum_{j \in I_{inef}}D^{\pi^*}(T_j +1) 2^{j-1}\\
        &\stackrel{\phantom{\text{Lem. } \ref{lem:alt_lower_bound}}}{\leq} \frac{\mathcal{P}}{3 \mathcal{D}} \sum_{j=0}^s D^{\pi^*}(T_j +1) 2^{j-1}\\
        &\stackrel{\text{Lem. } \ref{lem:alt_lower_bound}}{\leq} \frac{\mathcal{P}}{3 \mathcal{D}} \cost(\pi^*)\\
        &\stackrel{\phantom{\text{Lem. } \ref{lem:alt_lower_bound}}}{\leq} \frac{\mathcal{P}}{3}
    \end{align*}

    Now we consider the blocked intervals. Let $j \in I_{block}$. Since $j \leq i-1$ it holds that $P_j \subseteq P_i$. Thus, the interval $[T_{j}+1,T_{j+1}]$ was also efficient after $P_j$ got collected and by Lemma \ref{lem:efficient_subtour}:
    \begin{equation*}
        p(\subtour_j) \geq \frac{p(A_j \setminus P_j)}{3+ \epsilon} \geq \frac{p(A_j \setminus P_i)}{3 + \epsilon}.
    \end{equation*}
    By summing up over all blocked intervals we obtain:
    \begin{align*}
        \sum_{j \in I_{block}} p(A_j \setminus P_i) &\leq (3 + \epsilon) \sum_{j \in I_{block}} p(\subtour_j)\\
        &\leq (3 + \epsilon) \sum_{j=0}^{i-1} p(\subtour_j)\\
        &= (3 + \epsilon) \cdot p(P_i)\\
        &< (3 + \epsilon) \cdot \frac{\mathcal{P}}{9+3\epsilon}\\
        &= \frac{\mathcal{P}}{3}
    \end{align*}

    By combining all of these bounds we may conclude that $I_{eff}$ cannot be empty:
    \begin{align*}
        \sum_{j \in I_{eff}} p(A_j \setminus P_i) &> \frac{2}{3} \mathcal{P} - \sum_{j \in I_{inef}}p(A_j \setminus P_i) - \sum_{j \in I_{block}} p(A_j \setminus P_i)\\
        &> \frac{2}{3}\mathcal{P} -  \frac{\mathcal{P}}{3} - \frac{\mathcal{P}}{3} = 0.
    \end{align*}
    This directly proves the lemma.
\end{proof}

Since there do not exists any intervals with indices greater $s+1$, we can use this lemma to directly lower-bound the total profit collected by the intervals $\subtour_1,\ldots,\subtour_s$:

\begin{restatable}{corollary}{CorrPickedTOne}\label{cor:picked_tour1}
    After executing Algorithm \ref{Alg:profit_tour1}, it holds that $p(P_{s+1}) \geq \frac{ \mathcal{P}}{9 + 3 \epsilon}$ if $p(B) \leq \frac{2 + \epsilon}{9 + 3 \epsilon}\mathcal{P}$.
\end{restatable}

While this lower-bounds the total profit collected by the subtours, Algorithm \ref{Alg:profit_tour1} might not append all subtours to $\pi^{(1)}$. But, this can only happen if $\cost(\pi^{(1)}) \geq (6 +6\epsilon) \mathcal{D}$ and we can lower-bound the ratio between the profit and the cost of $\pi^{(1)}$.

\begin{restatable}{lemma}{LemTourOneProfit}\label{lem:tour1_profit}
    If $p(B) \leq \frac{2 + \epsilon}{9 + 3 \epsilon} \mathcal{P}$ it holds that $p(\pi^{(1)}) \geq \frac{1 }{9 + 3 \epsilon} \mathcal{P}$.
\end{restatable}

\begin{proof}
    We distinguish two cases. In the first one $\pi^{(1)}$ consists of all the subtours $\subtour_1,\ldots, \subtour_s$ appended to each other. By Corollary~\ref{cor:picked_tour1} we know that $p(P_{s+1}) \geq \frac{ \mathcal{P}}{9 + 3 \epsilon}$. Since $P_{s+1}$ is the union of the items that are collected on the subtours $\subtour_0,\ldots, \subtour_s$ this implies that $p(\pi^{(1)}) = p(P_{s+1}) \geq \frac{ \mathcal{P}}{9 + 3 \epsilon}$.
    
    In the second case, there exists an $i \in [s]$ such that the algorithm does not append $\subtour_i$ to $\pi^{(1)}$. This can only happen if $\cost(\pi^{(1)}) > (6 + 6 \epsilon) \mathcal{D}$ when the algorithm considers to append the subtour $\subtour_i$ to $\pi^{(1)}$. Then the tour $\pi^{(1)}$ consists of the subtours $\subtour_0,\ldots,\subtour_{i-1}$ that get traversed after each other. One can bound its travel time as follows:
    \begin{align*}
        \cost(\pi^{(1)}) &\leq \sum_{j=0}^{i-1} \cost\left(\subtour_j,w(P_i)\right)\\
        &\leq \sum_{j=0}^{i-1} \frac{p(\subtour_j)}{t_1}\\
        &= \frac{{p(\pi^{(1)}})}{t_1}.
    \end{align*}

    Thus, the profit of $\pi^{(1)}$ is lower bounded by its travel time times $t_1$ and since its travel time greater than $(6 + 6 \epsilon) \mathcal{D}$ we obtain:
    \begin{equation*}
        p(\pi^{(1)}) > t_1 \cdot (6 + 6 \epsilon) \mathcal{D} = \frac{1}{9 + 3 \epsilon} \cdot \frac{\mathcal{P}}{(6 + 6\epsilon)\mathcal{D}} (6 + 6 \epsilon) \mathcal{D} = \frac{\mathcal{P}}{9 + 3\epsilon}.
    \end{equation*}
    Thus, the lemma is proven.
\end{proof}

Lastly, we have to prove that the cost of $\pi^{(1)}$ is not too large.

\begin{restatable}{lemma}{LemCostTourOne}
    It holds that $\cost(\pi^{(1)}) \leq (9 + 9\epsilon) \mathcal{D}$.
\end{restatable}

\begin{proof}
    Let $i \in [s]$ be the index of the last tour that gets appended to $\pi^{(1)}$. Let $\sigma^{(1)}$ denote the tour $\pi^{(1)}$ before $\subtour_i$ gets appended. Since Algorithm~\ref{Alg:profit_tour1} only continues to append subtours as long as $\cost(\pi^{(1)}) \leq (6 + 6\epsilon)\mathcal{D}$ we know that $\cost(\sigma_1) \leq (6 + 6\epsilon)\mathcal{D}$. We can bound the cost of $\pi^{(1)}$ as follows:
    \begin{equation*}
        \cost(\pi^{(1)}) \leq \cost(\sigma^{(1)}) + \cost(\subtour_i,w(\sigma^{(1)})) \leq (6 + 6 \epsilon) \mathcal{P} + \cost(\subtour_i,w(\sigma^{(1)})).
    \end{equation*}
    We know that the weight of the items picked up on the tours $\subtour_1,\ldots,\subtour_{i-1}$ is upper bounded by $T_i$ and that the weight picked up tour $\subtour_i$ is upper bounded by $T_{i+1} - T_i$. Thus, at the beginning of the subtour $\subtour_i$ the collected weight is upper bounded by $T_i$ and at the end it is upper bounded by $T_{i+1}$. Additionally, we know that the length of the path after the agent collects the first item on the subtour $\subtour_i$ is upper bounded by $(1+ \epsilon)^{\lceil\log_{1+\epsilon}(\mathcal{D}/2^{i})\rceil} \leq (1 + \epsilon) \frac{\mathcal{D}}{2^i}$. By the triangle inequality we may also bound the distance from the $r$ to the location of the first item by this value and we obtain:
    \begin{align*}
        \cost(\subtour_i,w(\sigma^{(1)})) &\leq (1 + \epsilon) \frac{\mathcal{D}}{2^i} f(T_i) + (1 + \epsilon) \frac{\mathcal{D}}{2^i} f(T_{i+1})\\
        &\leq (1 + \epsilon) \frac{\mathcal{D}}{2^i} (2^i + 2^{i+1})\\
        &\leq (3 + 3\epsilon) \mathcal{D}
    \end{align*}
    Thus, $\cost(\pi^{(1)}) \leq (6 + 6\epsilon) \mathcal{D} + (3 + 3\epsilon) \mathcal{D} = (9 + 9\epsilon) \mathcal{D}$ and the lemma is proven.
\end{proof}

\subsection{Obtaining the Second Solution}\label{sec:path2}

We describe an algorithm that calculates a tour $\pi^{(2)}$ whose profit can be lower bounded in terms of the profit of the items in $p(B)$. The tour basically consists of subtours on which the agent collects a single item and returns to $r$ in between. The algorithm will only collect an item (after picking up a certain weight) when the cost of doing this is upper bounded by a threshold $t_2$ times the profit of the item. Since the agent only gets slower over time, we can define for any item $v \in V$ a weight $\lborder(v)$ such that we are willing to pick up $v$ if and only if the already collected weight is at most $\lborder(v)$. Additionally, we also require that $\cost((r,v,r),s(v)) \leq 2 \mathcal{D}$ to prevent that a single subtour is too expensive. Our algorithm maximizes the profit of the collected items such that no item $v$ is collected after the agent already collected items with total weight greater $\lborder(v)$. This can be modeled as the \emph{Multiperiod Binary Knapsack Problem} (\knapdead) which is a knapsack variant with deadlines~\cite{GaoBG21}:

\begin{restatable}{definition}{defMPBKP}
    In the \emph{Multiperiod Binary Knapsack Problem} (\knapdead) one is given a set of items $V$, a weight function $w:V \rightarrow \mathbb{N}_{0}$, a profit or reward function $p:V \rightarrow \mathbb{N}_{0}$ and a function $s: V \rightarrow \mathbb{N}_{ 0}$ that assigns a deadline to every item. The aim is to find a tuple $\sigma = (\sigma_1,\ldots,\sigma_l)$ of distinct items such that the profit $p(\sigma) = \sum_{i=1}^l p(\sigma_i)$ gets maximized and for any $i \in [l]$ it holds that $\sum_{i=1}^{l-1} w(\sigma_i) \leq \lborder(\sigma_l)$.
\end{restatable}

The definition differs from the one presented by Gao et al.~\cite{GaoBG21} but it is easy to verify that both definitions are  equivalent. They provided a fully polynomial $(1+\epsilon)$-approximation scheme for this problem. In our algorithm we use this scheme as a function $\knapdead(V,p,w,s)$ to calculate an $(1 +\epsilon)$ approximation. After calculating the respective solution, we simply append items to the initially empty tour $\pi^{(2)}$ until $\cost(\pi^{(2)})$ exceeds $6 \mathcal{D}$ or no items remain. Algorithm~\ref{Alg:profit_tour2} provides the pseudocode of this procedure.

\begin{restatable}{algorithm}{AlgTourTwo}
	\LinesNumbered
	\DontPrintSemicolon
	\SetKwInOut{Input}{input}
	\SetKwInOut{Output}{output}
    \caption{Calculation of the second tour $\pi^{(2)}$.}\label{Alg:profit_tour2}
        $t_2 = \frac{ 6 \cdot (9 + 3\epsilon) \cdot \mathcal{D}}{\mathcal{P}}$\;
        Define the function $\lborder: V \rightarrow \mathbb{N}_{0}$\;
        \For{ $v \in V \setminus \{r\}$}{
            Calculate the value $\lborder(v)$ such that $\cost((r,v,r),\lborder(v)) \leq \min(p(v) \cdot t_2, 2\mathcal{D}) < \cost((r,v,r),\lborder(v) +1)$\;
        }
        $(\sigma_{1},\ldots,\sigma_l)=\knapdead(V,p,w,\lborder)$\;
        Find minimum $l' \in [l]$ such that $\cost(r,\sigma_1,\ldots,\sigma_{l'},r) \geq 6 \mathcal{D}$.\;
        If no such $l'$ exists set $l' = l$.\;
        $\pi^{(2)} = (r,\sigma_1,\ldots,\sigma_{l'},r)$\;
        \textbf{return:} $\pi^{(2)}$
\end{restatable}


Similarly to the analysis of Algorithm \ref{Alg:profit_tour1}, we subdivide the items in $B$ into two groups, efficient and inefficient items:

\begin{restatable}{definition}{DefEffTourTwo}
    Let $\beta_1,\ldots,\beta_m$ be the elements of $B$ ordered by the moment in which the agent visits the respective location in $\pi^*$. For any $j \in [m]$ let $W^*_{\beta_j}$ be the weight that the agent has collected on $\pi^*$ once it reaches $\beta_j$. We define $i_j$ as the largest index such that $W^*_{\beta_j} \geq T_{i_j} +1$. We call $\beta_j$ \emph{efficient} if $\frac{p(\beta_j)}{D^{\pi^*}(T_{i_j}+1)\cdot 2^{i_j-1}} \cdot (9 + 3 \epsilon) \geq \frac{\mathcal{P}}{\mathcal{D}}$.
\end{restatable}

By the definition of $B$, we know for any $\beta_j \in B$ that the weight collected before and after visiting $\beta_j$ cannot be in the same interval $[T_i+1,T_{i+1}]$ for an $i \in [s] \cup \{0\}$. Because of this, the indices $i_1,\ldots,i_m$ need to be strictly increasing.

\begin{restatable}{observation}{ObsStepsB}\label{obs:steps_b}
    For any $j \in [m-1]$ it holds that $i_{j} < i_{j +1}$.
\end{restatable}

One can show that one can pack the efficient items in $B$ into the knapsack without violating the respective deadlines by using the order in which they are visited on the optimum tour $\pi^*$.

\begin{restatable}{lemma}{LemTourTwoEff}\label{lem:tour2_efficient}
    For any $j \in [m]$ it holds that $\lborder(\beta_j) \geq \sum_{j' = 1}^{j-1} w(\beta_{j'})$ if $\beta_j$ is efficient.
\end{restatable}

\begin{proof}
    Let $W^\beta_{j-1} = \sum_{j' = 1}^{j-1} w(\beta_{j'})$. It is easy to verify that $W^\beta_{j-1} + w(\beta_j) \leq T_{i_j +1}$. Additionally by Observation~\ref{obs:steps_b} we know that $i_{j-1} \leq i_j-1$ and thus $W^\beta_{j-1} \leq T_{i_j}$. Thus
    \begin{align*}
        \cost((r,\beta_j,r),W^\beta_{j-1}) &= d(r,\beta_j) \cdot f\left(W_{j-1}^\beta\right) + d(\beta_j,r) \cdot f\left(W_{j-1}^\beta + w(\beta_j)\right)\\
        &\leq d(r,\beta_j) \cdot \left( f(T_{i_j}) + f(T_{i_j +1})\right)\\
        &\leq d(r,\beta_j) \cdot (2^{i_j} + 2^{i_j +1})\\
        &\leq 6 \cdot (D^{\pi^*}(T_{i_j}+1)\cdot 2^{i_j-1})\\
        &\leq 6 \cdot\left(p(\beta_j) \cdot (9 + 3\epsilon) \cdot \frac{\mathcal{D}}{\mathcal{P}}\right)\\
        &= p(\beta_j)  \cdot t_2.
    \end{align*}

    Additionally, we know that on the optimum tour the agent travels at least distance $D^{\pi^*}(W^\beta_{j-1} + \beta_j) \geq d(r,\beta_j)$ after picking up weight $W^\beta_{j-1} + \beta_j$. Thus,
    \begin{equation*}
        \cost((r,\beta_j,r),W^\beta_{j-1}) \leq 2 f(W^\beta_{j-1} + \beta_j) \cdot  D^{\pi^*}(W^\beta_{j-1} + \beta_j) \leq 2 \mathcal{D}.
    \end{equation*}
    By combining these two bounds we obtain $\cost((r,\beta_j,r),W^\beta_{j-1}) \leq \min(p(\beta_j) \cdot t_2, 2\mathcal{D})$ which implies that $\lborder(\beta_j) \geq W^\beta_{j-1}$.
\end{proof}

This implies that the efficient items in $B$ form a feasible $\knapdead$ solution. By bounding the total profit of the inefficient items in $B$, we obtain a lower bound for the profit of $\sigma$:

\begin{restatable}{lemma}{LemTourTwoKpSol}\label{lem:tour2_kp_sol}
    If $p(B) \geq \frac{2+ \epsilon}{9+3\epsilon} \mathcal{P}$ then it holds that $p(\sigma) \geq \frac{\mathcal{P}}{9+3\epsilon}$.
\end{restatable}

\begin{proof}
    Let us consider the $\knapdead$ solution $(\beta_{j_1},\ldots,\beta_{j_{m'}})$ such that $j_1,\ldots,j_{m'}$ are exactly the indices of the efficient items in $B$ in the order in which they are visited on the optimum tour. One may verify that this is a feasible solution because for any $z \in [m']$ it holds that
    \begin{equation*}
        \sum_{z' = 1}^{z-1} w(\beta_{j_{z'}}) \leq \sum_{j' = 1}^{j_z-1} w(\beta_{j'}) \stackrel{\text{Lem. } \ref{lem:tour2_efficient}}{\leq} \lborder(\beta_{j_{z}}).
    \end{equation*}
    It remains to lower bound the profit of this solution. To do this let $I_{inef}$be a set containing the indices all inefficient items in $B$. Then
    \begin{align*}
        \sum_{j \in I_{inef}} p(\beta_j) &\stackrel{\phantom{\text{Lem. } \ref{lem:alt_lower_bound}}}{<} \frac{\mathcal{P}}{(9 + 3 \epsilon) \mathcal{D}}\sum_{j \in I_{inef}} D^{\pi^*}(T_{i_j}+1) \cdot 2^{i_j-1}\\
        &\stackrel{\text{Obs. } \ref{obs:steps_b}}{\leq} \frac{\mathcal{P}}{{(9 + 3 \epsilon) \mathcal{D}}} \sum_{i=0}^s D^{\pi^*}(T_i +1) \cdot 2^{i -1}\\
        &\stackrel{\text{Lem. } \ref{lem:alt_lower_bound}}{\leq} \frac{\mathcal{P}}{{(9 + 3 \epsilon) \mathcal{D}}} \cost(\pi^*)\\
        &\stackrel{\phantom{\text{Lem. } \ref{lem:alt_lower_bound}}}{\leq} \frac{\mathcal{P}}{9 + 3 \epsilon}
    \end{align*}

    Thus, the total profit of the efficient items can be lower bounded by
    \begin{equation*}
        p(B) - \sum_{j \in I_{inef}} p(\beta_j) > \frac{2+ \epsilon}{9+3\epsilon} \mathcal{P} - \frac{\mathcal{P}}{9 + 3 \epsilon} = \frac{1+\epsilon}{9 + 3 \epsilon}.
    \end{equation*}

    Since $\sigma$ is an $(1+\epsilon)$-approximation of the optimum $\knapdead$ solution, this implies that $p(\sigma) \geq \frac{\mathcal{P}}{9+3\epsilon}$.
\end{proof}

As in the proof of Lemma \ref{lem:tour1_profit}, we can show that if $\pi^{(2)}$ does not collect all items contained in $\sigma$ then the profit of $\pi^{(2)}$ can be lower bounded because the ratio between its cost and profit is bounded. If $\pi^{(2)}$ collects all items in $\sigma$, we can apply the previous lemma to lower bound the profit of $\pi^{(2)}$.

\begin{restatable}{lemma}{LemProfitTourTwo}\label{lem:profit_tour2}
    If $p(B) \geq \frac{2+ \epsilon}{9+3\epsilon} \mathcal{P}$ then it holds that $p(\pi^{(2)}) \geq \frac{\mathcal{P}}{9+3\epsilon}$.
\end{restatable}

\begin{proof} 
    By Lemma~\ref{lem:tour2_kp_sol}, we know that $p(\sigma) \geq \frac{\mathcal{P}}{9+3\epsilon}$. Thus, it can only happen that $p(\pi^{(2)}) < \frac{\mathcal{P}}{9+3\epsilon}$ if there exists an $l' \in [l -1]$ such that $\pi^{(2)} = (r,\sigma_1,\ldots,\sigma_{l'})$ and $\cost(\pi^{(2)}) > 6 \mathcal{D}$. We may then lower-bound $p(\pi^{(2)})$ as follows:
    \begin{align*}
        p(\pi^{(2)}) &= \sum_{j=1}^{l'} p(\sigma_j)\\
        &\geq \sum_{j=1}^{l'} \frac{1}{t_2}\cdot \cost((r,\sigma_j,r),\lborder(\sigma_j))\\
        &\geq \frac{1}{t_2} \sum_{j=1}^{l'} \cost\left((r,\sigma_j,r),\sum_{j' = 1}^{j-1}w(\sigma_{j'})\right)\\
        &\geq \frac{1}{t_2} \cdot \cost(\pi^{(2)})\\
        &> \frac{\mathcal{P}}{6\cdot (9+3\epsilon)\cdot \mathcal{D}} \cdot 6 \mathcal{D}\\
        &=\frac{\mathcal{P}}{9 + 3 \epsilon}.
    \end{align*}
    This proves the lemma.
\end{proof}

Lastly, we bound the travel time of the tour $\pi^{(2)}$:

\begin{restatable}{lemma}{LemCostTourTwo}
    It holds that $\cost(\pi^{(2)})\leq 8 \mathcal{D}$.
\end{restatable}

\begin{proof}
    Let $l'$ be chosen such that $\pi^{(2)} = (r,\sigma_1,\ldots,\sigma_{l'},r)$. Then it holds that $\cost(r,\sigma_1,\ldots,\sigma_{l'-1},r) \leq 6 \mathcal{D}$ and we may bound $\cost(\pi^{(2)})$ as follows:
    \begin{align*}
        \cost(\pi^{(2)}) &\leq  6 \mathcal{D} + \cost\left((r,\sigma_{l'},r),\sum_{j=1}^{l'-1} w(\sigma_j)\right)\\
        &\leq 6 \mathcal{D} + \cost((r,\sigma_{l'},r),\lborder(\sigma_{l'}))\\
        &\leq 6 \mathcal{D} + 2\mathcal{D} \\
        &= 8 \mathcal{D},
    \end{align*}
    where we used in the last inequality that $\lborder(\sigma_{l'})$ is chosen in such a way that $\cost((r,\sigma_{l'},r),\lborder(\sigma_{l'})) \leq 2 \mathcal{D}$.
\end{proof}

By executing both Algorithm~\ref{Alg:profit_tour1} and Algorithm~\ref{Alg:profit_tour2} and taking the tour with the larger profit, we obtain a bicriteria approximation for \cthief:

\begin{restatable}{theorem}{ThmCTTP}
    For given $\mathcal{P}\in \mathbb{N}$ and an $\epsilon>0$ one can compute in polynomial time a $\cthief$ tour picking up weight at least $\frac{1}{9 + \epsilon} \mathcal{P}$ whose travel cost is upper bounded by $(9+\epsilon)$ times the optimum travel cost of any tour picking up weight at least $\mathcal{P}$.
\end{restatable}

\bibliography{sample}

@inproceedings{epstein2010universal,
  title={Universal sequencing on a single machine},
  author={Epstein, Leah and Levin, Asaf and Marchetti-Spaccamela, Alberto and Megow, Nicole and Mestre, Juli{\'a}n and Skutella, Martin and Stougie, Leen},
  booktitle={International Conference on Integer Programming and Combinatorial Optimization},
  pages={230--243},
  year={2010},
  organization={Springer}
}

@inproceedings{GaoBG21,
  author       = {Zuguang Gao and
                  John R. Birge and
                  Varun Gupta},
  title        = {Approximation Schemes for Multiperiod Binary Knapsack Problems},
  booktitle    = {Computer Science - Theory and Applications - 16th International Computer
                  Science Symposium in Russia, {CSR} 2021, Sochi, Russia, June 28 -
                  July 2, 2021, Proceedings},
  doi          = {10.1007/978-3-030-79416-3\_8}
}

@incollection{ausiello2018prize,
  title={Prize collecting traveling salesman and related problems},
  author={Ausiello, Giorgio and Bonifaci, Vincenzo and Leonardi, Stefano and Marchetti-Spaccamela, Alberto},
  booktitle={Handbook of Approximation Algorithms and Metaheuristics},
  pages={611--628},
  year={2018},
  publisher={Chapman and Hall/CRC}
}

@article{orienteering2apx,
    author = {Chekuri, Chandra and Korula, Nitish and P\'{a}l, Martin},
    title = {Improved algorithms for orienteering and related problems},
    year = {2012},
    journal={ACM Transactions on Algorithms (TALG)},
    volume = {8},
    number = {3},
    issn = {1549-6325},
    doi = {10.1145/2229163.2229167}
    }

@article{teamorienteering,
title = "Approximation Algorithms for the Team Orienteering Problem",
author = "Wenzheng Xu and Zichuan Xu and Jian Peng and Weifa Liang and Tang Liu and Xiaohua Jia and Das, \{Sajal K.\}",
year = "2020",
doi = "10.1109/INFOCOM41043.2020.9155343",
journal = "Proceedings - IEEE INFOCOM",
}

@article{bock2015capacitated,
  title={The capacitated orienteering problem},
  author={Bock, Adrian and Sanit{\`a}, Laura},
  journal={Discrete Applied Mathematics},
  volume={195},
  pages={31--42},
  year={2015},
  publisher={Elsevier}
}

@inproceedings{chaudhuri2003paths,
  title={Paths, trees, and minimum latency tours},
  author={Chaudhuri, Kamalika and Godfrey, Brighten and Rao, Satish and Talwar, Kunal},
  booktitle={44th Annual IEEE Symposium on Foundations of Computer Science, 2003. Proceedings.},
  pages={36--45},
  year={2003},
  organization={IEEE}
}

@article{arora20062+,
  title={A 2+$\epsilon$ approximation algorithm for the k-{MST} problem},
  author={Arora, Sanjeev and Karakostas, George},
  journal={Mathematical Programming},
  volume={107},
  number={3},
  pages={491--504},
  year={2006},
  publisher={Springer}
}

@inproceedings{blank2017solving,
  title={Solving the bi-objective traveling thief problem with multi-objective evolutionary algorithms},
  author={Blank, Julian and Deb, Kalyanmoy and Mostaghim, Sanaz},
  booktitle={International Conference on Evolutionary Multi-Criterion Optimization},
  pages={46--60},
  year={2017},
  organization={Springer},
  doi = {10.1007/978-3-319-54157-0_4}
}

@article{CHAGAS2022105560,
title = {A weighted-sum method for solving the bi-objective traveling thief problem},
journal = {Computers \& Operations Research},
volume = {138},
pages = {105560},
year = {2022},
doi = {10.1016/j.cor.2021.105560},
author = {Jonatas B.C. Chagas and Markus Wagner}
}

@inproceedings{DBLP:conf/gecco/PolyakovskiyB0MN14,
  author       = {Sergey Polyakovskiy and
                  Mohammad Reza Bonyadi and
                  Markus Wagner and
                  Zbigniew Michalewicz and
                  Frank Neumann},
  title        = {A comprehensive benchmark set and heuristics for the traveling thief
                  problem},
  booktitle    = {Genetic and Evolutionary Computation Conference, {GECCO} 2014},
  pages        = {477--484},
  publisher    = {{ACM}},
  year         = {2014},
  url          = {https://doi.org/10.1145/2576768.2598249},
  doi          = {10.1145/2576768.2598249},
  bibsource    = {dblp computer science bibliography, https://dblp.org}
}

@article{DBLP:journals/ec/PrzybylekWM18,
  author       = {Michal R. Przybylek and
                  Adam Wierzbicki and
                  Zbigniew Michalewicz},
  title        = {Decomposition Algorithms for a Multi-Hard Problem},
  journal      = {Evol. Comput.},
  volume       = {26},
  number       = {3},
  year         = {2018}
}

@article{DBLP:journals/swevo/HerringKY24,
  author       = {Daniel Herring and
                  Michael Kirley and
                  Xin Yao},
  title        = {A comparative study of evolutionary approaches to the bi-objective
                  dynamic Travelling Thief Problem},
  journal      = {Swarm Evol. Comput.},
  volume       = {84},
  pages        = {101433},
  year         = {2024}
}

@article{DBLP:journals/telo/NikfarjamNN24,
  author       = {Adel Nikfarjam and
                  Aneta Neumann and
                  Frank Neumann},
  title        = {On the Use of Quality Diversity Algorithms for the Travelling Thief
                  Problem},
  journal      = {{ACM} Trans. Evol. Learn. Optim.},
  volume       = {4},
  number       = {2},
  pages        = {12},
  year         = {2024}
}

@inproceedings{DBLP:conf/algocloud/NeumannPSSW18,
  author       = {Frank Neumann and
                  Sergey Polyakovskiy and
                  Martin Skutella and
                  Leen Stougie and
                  Junhua Wu},
  title        = {A Fully Polynomial Time Approximation Scheme for Packing While Traveling},
  booktitle    = {{ALGOCLOUD}},
  series       = {Lecture Notes in Computer Science},
  volume       = {11409},
  pages        = {59--72},
  publisher    = {Springer},
  year         = {2018}
}

@inproceedings{DBLP:conf/gecco/BossekCK020,
  author    = {Jakob Bossek and
               Katrin Casel and
               Pascal Kerschke and
               Frank Neumann},
  title     = {The node weight dependent traveling salesperson problem: approximation
               algorithms and randomized search heuristics},
   booktitle    = {Genetic and Evolutionary Computation Conference, {GECCO} 2020},
  pages     = {1286--1294},
  publisher = {{ACM}},
  year      = {2020}
}

@inproceedings{yafrani2017multi,
  title={Multi-objectiveness in the single-objective traveling thief problem},
  author={Yafrani, Mohamed El and Chand, Shelvin and Neumann, Aneta and Ahiod, Bela{\"\i}d and Wagner, Markus},
  booktitle={Proceedings of the Genetic and Evolutionary Computation Conference Companion},
  pages={107--108},
  year={2017}
}

@inproceedings{wu2018evolutionary,
  title={Evolutionary computation plus dynamic programming for the bi-objective travelling thief problem},
  author={Wu, Junhua and Polyakovskiy, Sergey and Wagner, Markus and Neumann, Frank},
  booktitle={Proceedings of the genetic and evolutionary computation conference},
  pages={777--784},
  year={2018}
}

@INPROCEEDINGS{Bonyadi2013TTP, 
author={M. R. {Bonyadi} and Z. {Michalewicz} and L. {Barone}}, 
booktitle={2013 IEEE Congress on Evolutionary Computation}, 
title={The travelling thief problem: The first step in the transition from theoretical problems to realistic problems}, 
year={2013}, 
volume={}, 
number={}, 
pages={1037-1044}, 
doi={10.1109/CEC.2013.6557681}, 
ISSN={1941-0026}}

@misc{eube2026effectivetravelingmetricinstances,
      title={Effective Traveling for Metric Instances of the Traveling Thief Problem}, 
      author={Jan Eube and Kelin Luo and Aneta Neumann and Frank Neumann and Heiko Röglin},
      year={2026},
      eprint={2604.19271},
      archivePrefix={arXiv},
      primaryClass={cs.DS},
      url={https://arxiv.org/abs/2604.19271}, 
}

\appendix

\section{A $(2e + \epsilon)$-approximation Algorithm for Weighted TSP}\label{apx:improved_weighted}

In this section we describe how the approximation factor for $\WTSP$ can be improved to $(2 e + \epsilon)$. First, we provide a randomized algorithm with this guarantee (following the same ideas as the randomized algorithm for universal sequencing \cite{epstein2010universal}). Then we describe how this algorithm can be derandomized. First, we generalize Lemma~\ref{lem:bound} for a randomized setting:

\begin{lemma}\label{lem:bound_rand}
    Let a randomized algorithm be given that produces a TSP tour $\pi$ and let $\pi^*$ be the optimum $\WTSP$ tour. Let $c$ be an arbitrary constant. If it holds for any weight $W$ that $E\left[D^{\pi}(W)\right] \leq c \cdot D^{\pi^*}(W)$ then $E\left[cost(\pi)\right] \leq c \cdot cost(\pi^*)$.
\end{lemma}

\begin{proof}

    Let us for simplicity define $f(-1) = 0$. For any tour $\pi$ we can express $\cost(\pi)$ as follows:
    \begin{equation*}
        \cost(\pi) = \sum_{W= 0}^{w(V)} D^{\pi^*}(W) \cdot (f(W) - f(W-1))
    \end{equation*}
    Using the linearity of expectation we obtain:
    \begin{align*}
        E[\cost(\pi)] &= E\left[ \sum_{W= 0}^{w(V)} D^\pi(W) \cdot (f(W) - f(W-1)) \right]\\
        &= \sum_{W= 0}^{w(V)} E[D^\pi(W)] \cdot (f(W) - f(W-1))\\
        &\leq \sum_{W= 0}^{w(V)} c \cdot D^{\pi^*}(W) \cdot (f(W) - f(W-1))\\
        &= c \cdot \cost(\pi^*)\qedhere
    \end{align*}
\end{proof}

Using this, we can improve the approximation factor of Algorithm \ref{Alg:8_apx} by sampling a number $\lambda$ uniformly from the interval $[0,1)$ and modify  the choice of $w_s$ in the algorithm such that:
\begin{equation*}
    \ell(\textsc{Qtsp}(r, (R,d), w_s)) \leq e^{s + \lambda} < \ell(\textsc{Qtsp}(r,(R,d), w_{s}+1))
\end{equation*}

Similarly to Lemma \ref{lem:opt_plan_lower}, one can show that:

\begin{lemma}
    If $(2 + \epsilon)\cdot D^{\pi^*}(\mathcal{W}) \leq e^{j+\lambda}$ then $W_j^\subtour \geq w(v) - \mathcal{W} +1$.
\end{lemma}

Under this circumstances the value $D^{\pi^*}(\mathcal{W})$ can be bounded by the length of the tours $\subtour_1,\ldots,\subtour_j$. By summing up the respective lengths and applying the harmonic sum we obtain:

 \begin{lemma}\label{lem:rand_plan_upper}
 If $(2 + \epsilon)\cdot D^{\pi^*}(\mathcal{W}) \leq e^{j+\lambda}$ then $D^\pi(\mathcal{W}) \leq e^{j + \lambda} \cdot \frac{e}{e-1}$
 \end{lemma}

 This lemma can then be used to bound the expected approximation factor of the randomized algorithm.

 \begin{theorem}
It holds that $E\left[\cost(\pi)\right] \leq (2 + \epsilon) \cdot e \cdot \cost(\pi^*)$.
 \end{theorem}

 \begin{proof}

 Let $\mathcal{W} \in w(V)$ be an arbitrary weight and let $c = \log_{e}((2 + \epsilon)\cdot D^{\pi^*}(\mathcal{W}))$. We can then use Lemma~\ref{lem:rand_plan_upper} to bound the expected value of $D^{\pi}(\mathcal{W})$ as follows:

\begin{align*}
    E_{\lambda \sim U[0,1]}\left[D^{\pi}(\mathcal{W})\right] &\leq E_{\lambda \sim U[0,1]}\left[\min_{j \in \mathbb{N}_0: j + \lambda \geq c}e^{j+\lambda}\cdot \frac{e}{e-1}\right]\\
    &= e^c \cdot \frac{e}{e-1} E_{\lambda \sim U[0,1]}\left[\min_{j \in \mathbb{N}_0: j + \lambda \geq c}e^{j+\lambda}/e^c\right]\\
    &= (2 + \epsilon)\cdot D^{\pi^*}(\mathcal{W}) \cdot \frac{e}{e-1} \int_{0}^{1} e^{\lambda'} d \lambda'\\
    &= (2 + \epsilon) \cdot e \cdot D^{\pi^*}(\mathcal{W})
\end{align*}

Since this holds for every arbitrary weight $\mathcal{W} \in w(V)$ we can directly  apply Lemma~\ref{lem:bound_rand} and the theorem is proven.
\end{proof}

We can derandomize the algorithm by choosing an integer value $l \in \mathbb{N}$ and execute the randomized algorithm with every value $\lambda \in \left\{ \frac{0}{l}, \frac{1}{l},\ldots,\frac{l-1}{l}\right\}$ and taking the solution with the smallest cost. It is easy to verify that if we had chosen the value $\lambda$ from $ \left\{ \frac{0}{l}, \frac{1}{l},\ldots,\frac{l-1}{l}\right\}$ uniformly at random the expected cost of the solution can be bounded by $ (1+ \frac{1}{l}) \cdot (2 + \epsilon) \cdot e \cdot \cost(\pi^*)$. Thus, for at least one of the possible choices the cost of the tour must lie below this expected value. By using a smaller value $\epsilon' < \epsilon$ instead of $\epsilon$ in this algorithm and choosing a sufficiently large $l$ one obtains an $(2e + \epsilon)$-approximation guarantee:

\ThmWeighted*

\section{Applying our Approach to Different Models}

We introduced the $\cthief$ model in this paper because it simultaneously captures the general challenges of the traveling thief problem while it also simplifies writing a lot compared to existing models. In this section we explain how our $\cthief$ algorithm can be adapted to also approximate different settings. In particular, we argue that our algorithm can be used to calculate an approximate Pareto set for $\bithief$ and that this is also true if the starting node is not fixed, if there can be multiple items at the same location or if the agent is required to visit every location even if it does not want to collect any items there. 

\subsection{Bi-objective TTP}\label{apx:bittp}
In $\bithief$ we are not given a desired profit $\mathcal{P}$ and instead one aims to find a Pareto set with respect to the profit $p(\pi)$ and the travel time $\cost(\pi)$ of the tours. A tour $\pi$ is Pareto-optimal if there exists no tour $\pi'$ such that $p(\pi') \geq p(\pi)$ and $\cost(\pi') \leq \cost(\pi)$, with at least one of these inequalities being strict. The Pareto set is then the set of all Pareto-optimal tours. Since both objectives are $\NP$-hard, computing this set is also $\NP$-hard. We therefore restrict our attention to the computation of an approximate Pareto set. For two values $\alpha_1, \alpha_2 \in \mathbb{R}_{> 0}$, we say that a set $S$ of tours is an $(\alpha_1,\alpha_2)$-approximate Pareto set if for every feasible tour $\pi$ there exists a tour $\pi' \in S$ such that $\alpha_1 \cdot p(\pi') \geq p(\pi)$ and $\cost(\pi') \leq \alpha_2 \cdot \cost(\pi)$.

If we are given a bicriteria $(\alpha_1, \alpha_2)$-approximation algorithm for $\cthief$ and an $\epsilon > 0$, we can obtain a $((1 + \epsilon) \cdot \alpha_1,\alpha_2)$-approximate Pareto set for $\bithief$ by simply applying the bicriteria approximation algorithm with $\mathcal{P} = (1+ \epsilon)^i$ for every $i \in \mathbb{N}_0$ with $i \leq \lceil \log_{1+\epsilon}(p(V))\rceil$, where $p(V) = \sum_{v \in V} p(v)$. One may verify that for a constant $\epsilon > 0$ the choices of $i$ are polynomial in the input size. Thus, we can use our $\cthief$ algorithm to calculate a $(9 + \epsilon, 9 + \epsilon)$-approximate Pareto set for $\bithief$.

\subsection{Visiting every Location}\label{apx:all_location}
In most of the existing literature the agent is required to visit every single location, even if it does not collect an item there. One can deal with this variant by using Christofides $1.5$–approximation algorithm for TSP~\cite{chaudhuri2003paths} to calculate a tour visiting all locations where the agent does not collect any items and adding this to the start of the $\cthief$ tour produced by our algorithm. Since the agent still travels with its maximum speed while visiting these locations, the travel time of this initial subtour is bounded by $1.5$ times the travel time of any solution that visits every node and the approximation factor $\alpha_2$ gets increased by at most $1.5$.

However, by being a little bit more careful we can also ensure that all locations are visited without increasing the approximation factor. To do this we require a stronger version of Lemma~\ref{lem:alt_lower_bound}:

\begin{lemma}\label{lem:apx_stronger_bound}
Let $\pi^*$ be the optimum tour. It holds that
    \begin{equation*}
        \cost(\pi^*) \geq \textcolor{red}{\frac{1}{2}\cdot D^{\pi^*}(0) +}\sum_{i=0}^{s} 2^{i-1} \cdot D^{\pi^*}(T_{i} + 1).
    \end{equation*}
\end{lemma}
One can get an intuition why also this stronger lemma is true by having a look at Figure~\ref{fig:lower_bound}. 

Let $\tau$ be the TSP tour calculated by Christofides algorithm. Since we assume that $f(0) = 1$ and $\pi^*$ needs to be a TSP tour if we require the agent to visit every node, we know that
\begin{equation*}
\frac{1}{2}\cdot D^{\pi^*}(0) = \frac{1}{2} f(0) \cdot \ell(\pi^*) \geq \frac{1}{3} \cdot \ell(\tau).
\end{equation*}

Because of Lemma~\ref{lem:apx_stronger_bound} we can use $\mathcal{D} - \frac{1}{3} \cdot \ell(\tau)$ instead of $\mathcal{D}$ in the definitions of the efficient intervals, the efficient items and the thresholds $t_1$, $t_2$. We can again lower bound $\max(p(\pi^{(1)},p(\pi^{(2)})$ by $\frac{\mathcal{P}}{9 + 3 \epsilon}$. If we ignore the cost of the subtour $\tau$, the cost of each tour is at most $(9 + 9 \epsilon)\mathcal{D} - 2 \ell(\tau)$. Since traveling the tour $\tau$ before collecting any item costs only $\ell(\tau)$, the travel time of the entire tours $\pi_1$ and $\pi_2$ is then at most $(9 + 9 \epsilon)\mathcal{D}$. The exact details of this analysis are omitted.

\subsection{Unknown Starting Location}\label{apx:no_start}

If the initial position $r$ of the agent is not given, we can simply evaluate all $n$ choices of $r$ and choose the best one. Since our algorithm requires that there is no item placed at $r$ we have to take care of the respective item in advance. We can simply run the algorithm for both the case that the item gets picked up and that we do not pick up the item and take the better solution. There are two models in the literature regarding the point in time at which the item at the starting position needs to be collected (if we decide to collect it). If the agent is required to pick up the item in the beginning, we can simply modify $f$ and $w_{\max}$ to incorporate that the agent already starts with an initial weight. If the item can be picked up in the end, we only have to modify $w_{\max}$ to make sure that the agent can still collect the item when finishing the tour.

Doing this we are able to obtain an approximate $\cthief$ solution. If we want to calculate an approximate Pareto set for $\bithief$ we need to calculate the respective Pareto sets for every possible starting node and then merge them to obtain an approximate Pareto set for the situation that $r$ is not fixed.

\subsection{Having Multiple Items at the Same Location}\label{apx:mult_items}

If multiple items are placed at the same location we can simply create multiple copies of the respective node and place exactly one item at each of these nodes. It could happen that the solution calculated by our algorithm visits different copies of a location at different points in time. In the classic TTP model one is required to visit every location at most once and needs to decide then which items are collected. We can handle this situation by modifying our tour such that we only visit a location when we visit the last copy of it. Once we visit it we collect all items that are placed at the copies visited by our original tour. This only causes some of the weights to be collected later which cannot increase the travel time.

However, we assume that the minimum distance between $r$ and any other location is at least $1$. If there are multiple items placed at the starting location, we would need to have copies of $r$ that are at distance $0$ to it. If we are allowed to collect the items at the starting location at the end of the tour this does not cause too much problems. We simply have to modify Algorithm~\ref{Alg:profit_tour1} to also consider $\subtour' = \textsc{Cop}((V \setminus P_i,d),0,T_{i+1} - T_i)$ as a possible choice for subtour $\subtour_i$ for every $i \in \{0,\ldots,s\}$. We also have to be a little bit more careful during some parts of our analysis, but this is mainly due to the fact that we want to avoid to divide by $0$. In principle the algorithm works as intended.

If we are required to collect the items placed at $r$ in the beginning of the tour, things get a little bit more complicated. For a given set $S'$ of the items that get collected at $r$, we can simply modify the function $f$ and the value of $w_{\max}$ to incorporate that the agent already starts with some weight. To do this, we need to be able to predetermine $S'$.

Let $o_1, \ldots, o_l$ be the items placed at $r$ and let them be ordered such that $\frac{p(o_1)}{w(0_1)} \geq \frac{p(o_2)}{w(0_2)}\geq \ldots \geq \frac{w(o_l)}{p(o_l)}$. For two indices $i,j \in [l]$ we define the set:
\begin{equation*}
    S_{i,j} = \{v_{i'} \mid i' \leq i \land w(o_{i'}) \leq w(o_j)\},
\end{equation*}
which contains exactly those items whose index is at most $i$ and whose weight is at most $w(o_j)$. For every $i,j \in [l]$, we apply our $\cthief$ algorithm with $S' = S_{i,j}$. Additionally, we also calculate a solution for the case that we only pick up item $o_i$ at $r$ for every $i \in [l]$ and a solution for $S' = \emptyset$. Among all the solutions we take then the one with the largest profit.

Let $S^*$ be the set of items that get collected at $r$ in the optimum solution and let $o_{j^*} = \argmax_{o_j \in S^*}  w(o_j)$ be the heaviest item in $S^*$. We consider the minimum index $i^*$ such that $w(S_{i^*,j^*}) \geq w(S^*)$ (such an index must exists because $S^* \subseteq S_{l,j^*}$). It is easy to verify that $p(S_{i^*,j^*}) \geq p(S^*)$. Additionally, $w(S_{i^*-1,j^*}) \leq w(S^*)$ and $w(o_{i^*}) \leq w(o_{j^*}) \leq w(S^*)$. Thus, we know that the $\cthief$ algorithm will be executed at least once with a set $S'$ such that $w(S') \leq w(S^*)$ and $p(S') \geq \frac{p(S^*)}{2}$. 
In this iteration, we are then also guaranteed to collect a $\frac{1}{9 + 3\epsilon}$ fraction of the profit of the items collected on the optimum tour $\pi^*$ at the other locations. In total, we collect at least a profit of $\frac{p(\pi^*)}{9 + 3 \epsilon}$ while we can also bound the length of the tour by $(9 + 9 \epsilon) \cdot \mathcal{D}$. Thus, we also have a bicriteria approximation algorithm for the case that multiple items can be placed at the same location, even if the items at $r$ can only be collected at the beginning of the tour.

\end{document}